\documentclass{article}
\usepackage[algoruled]{algorithm2e}
\usepackage[letterpaper]{geometry}
\usepackage[T1]{fontenc} \usepackage{ae,aecompl}
\usepackage[utf8]{inputenc} \usepackage{mdwlist}
\usepackage[small]{caption} \usepackage{cite} \usepackage{times}
\usepackage{amsfonts} \usepackage{amsmath} \usepackage{amssymb}
\usepackage{graphicx} \usepackage{subfig}
\usepackage{multirow} \usepackage{url}
\usepackage{booktabs} \usepackage{color, colortbl}
\usepackage{amsthm}
\definecolor{Gray}{gray}{0.9}
\newcolumntype{g}{>{\columncolor{Gray}}r}
 \newtheorem{fact}{Fact}
\newtheorem{lemma}{Lemma} 
\newtheorem{theorem}{Theorem} \newtheorem{corollary}{Corollary}
\newtheorem{definition}{Definition}
\def\MR{MR$(M_G,M_L)$} 
 \def\Pr{\mbox{Pr}}
\def\dist{\mbox{dist}}
\newcommand{\BO}[1]{{O}\left(#1\right)}
\newcommand{\BT}[1]{{\Theta}\left(#1\right)}
\newcommand{\BOM}[1]{{\Omega}\left(#1\right)}
 \newcommand*\Let[2]{#1 $\gets$
  #2}

\newcommand{\affaddr}[1]{{\small #1}}
\newcommand{\email}[1]{{\tt\small #1}}

\begin{document}
%

\title{Space and Time Efficient Parallel Graph Decomposition, Clustering, and Diameter Approximation}
%
%
%
%
%

%
\date{}
\author{
%
%
Matteo Ceccarello \\
       \affaddr{Department of Information Engineering}\\
       \affaddr{University of Padova}\\
       \affaddr{Padova, Italy}\\
       \email{ceccarel@dei.unipd.it}
\and
Andrea Pietracaprina \\
       \affaddr{Department of Information Engineering}\\
       \affaddr{University of Padova}\\
       \affaddr{Padova, Italy}\\
       \email{capri@dei.unipd.it}
\and
Geppino Pucci \\
       \affaddr{Department of Information Engineering}\\
       \affaddr{University of Padova}\\
       \affaddr{Padova, Italy}\\
       \email{geppo@dei.unipd.it}
       \and
Eli Upfal\\
       \affaddr{Department of Computer Science}\\
       \affaddr{Brown University}\\
       \affaddr{Providence, RI, USA}\\
       \email{eli\_upfal@brown.edu}
}

\maketitle

\begin{abstract}\sloppy
  We develop a novel parallel decomposition strategy for unweighted,
  undirected graphs, based on growing disjoint connected clusters from
  batches of centers progressively selected from yet uncovered
  nodes. With respect to similar previous decompositions, our strategy
  exercises a tighter control on both the number of clusters and their
  maximum radius.

  We present two important applications of our parallel graph
  decomposition: (1) $k$-center clustering approximation; and (2)
  diameter approximation. In both cases, we obtain algorithms which
  feature a polylogarithmic approximation factor and are amenable to a
  distributed implementation that is geared for massive
  (long-diameter) graphs.  The total space needed for the computation
  is linear in the problem size, and the parallel depth is
  substantially sublinear in the diameter for graphs with low doubling
  dimension. To the best of our knowledge, ours are the first parallel
  approximations for these problems which achieve sub-diameter
  parallel time, for a relevant class of graphs, using only linear
  space.  Besides the theoretical guarantees, our algorithms allow for
  a very simple implementation on clustered architectures: we report
  on extensive experiments which demonstrate their effectiveness and
  efficiency on large graphs as compared to alternative known
  approaches.
\end{abstract}

\section{Introduction}

\emph{Graph analytics} is rapidly establishing itself as a major
discovery tool in such diverse application domains as road systems,
social networks, natural language processing, biological pattern
discovery, cybersecurity, and more.  Graph analytics tasks for big
data networks are typically run on distributed architectures such as
clusters of loosely-coupled commodity servers, where the challenge is
to minimize communication overhead between processors, while each
processor can store only a small fraction of the entire network.  A
number of computational models proposed in recent years
\cite{KarloffSV10,GoodrichSZ11,PietracaprinaPRSU12}  provide
excellent rigorous frameworks for studying algorithms for  massive data,
subject to these constraints.  Under this new
{computational paradigm,}
state-of-the-art graph algorithms often do not scale up efficiently to
process massive instances, since they either require superlinear
memory or exhibit long critical paths resulting in a large number of
communication rounds.

In this work we focus on {\em graph decomposition}, which is a
fundamental primitive for graph analytics as well as for several other
application contexts, especially in distributed settings, where
decompositions are often at the base of efficient parallel solutions.
We develop an efficient parallel decomposition algorithm for
partitioning the nodes of an unweighted, undirected graph into
disjoint, internally connected \emph{clusters}, which is able to
control the maximum radius of the clusters (the maximum distance of a
node in a cluster to the cluster's center).  Similarly to other known
decomposition approaches, our algorithm grows clusters from several
batches of centers which are progressively selected from the uncovered
nodes. However, rather than fixing the radius of each grown cluster
\emph{a priori}, or randomly delaying the activation of the centers,
as in previous works, we activate a new batch of centers every time
that the number of uncovered nodes halves, while continuing growing
the clusters of previously activated centers. The idea behind such a
strategy is to force more clusters to grow in poorly connected regions
of the graph while keeping both the total number of clusters and the
maximum cluster radius under control.

\sloppy
{We demonstrate the quality and utility of our decomposition strategy by applying it to the solution of two
extensively studied  problems, metric $k$-center clustering and diameter approximation, for which we obtain  space and time
efficient parallel algorithms.}

{In the  \emph{metric $k$-center} clustering problem~\cite{Gonzalez85,HochbaumS87} the goal is to partition an undirected
graph into $k$ clusters so that the maximum radius of the clusters is minimized.
The problem is NP-hard and we are therefore interested in efficient approximations. Given an $n$ node
unweighted undirected graph, building on our parallel graph decomposition method, we obtain a randomized
$\BO{\log^3 n}$-approximation algorithm to the $k$-center problem.}
The algorithm can be
implemented on the MapReduce (MR) model of \cite{PietracaprinaPRSU12}
in a number of parallel rounds proportional to the maximum cluster
radius using overall linear space, as long as each processing node is
provided with $\BOM{n^\epsilon}$ local space, for any constant
$\epsilon > 0$.  In order to derive a more explicit bound on the
parallel complexity, we analyze the maximum cluster radius, hence the
number of rounds, as a function of the doubling dimension of the graph
(see Definition~\ref{doublingdim}), showing that for a graph of
diameter $\Delta$ and doubling dimension $b >0$ the algorithm can
provide a decomposition into $k$ clusters with maximum cluster radius
$\tilde{O}(\lceil \Delta/k^{1/b} \rceil)$.

{Next, we apply our graph decomposition strategy to a
challenging problem in the context of graph analytics, namely, the
approximation of the graph diameter, a global property of a graph, in a number
of parallel rounds which is substantially less than the diameter itself, and using
linear global space and local memory at each processor sufficient to
store only a small fraction of the graph.}  We remark that known
parallel approaches to estimating the diameter of a graph either require
up to quadratic space (e.g., using transitive closure) or require a
number of parallel rounds which is inherently linear in the diameter
(e.g., using straightforward parallelization of breadth-first search
or neighborhood function estimations).

To estimate the diameter of a graph $G$, we first compute a
decomposition of $G$ of suitable granularity with our novel algorithm,
and then we estimate the diameter through the diameter of the
\emph{quotient graph}, that is, the graph whose nodes correspond to
the clusters and whose edges connect clusters containing nodes which
are adjacent in $G$. The granularity is chosen so that the size of the
quotient graph is small enough so that its diameter can be computed
using limited local memory and very few communication rounds.

We show that on any unweighted, undirected connected graph $G$ with
$n$ nodes and $m$ edges, our algorithm returns an upper bound to its
diameter which is a factor $\BO{\log^3 n}$ away from the true value,
with high probability.  The algorithm can be implemented on the
aforementioned MR model using overall linear space and
$\BOM{n^\epsilon}$ local space, with $\epsilon >0$, in a number of
parallel rounds which is $\tilde{O}(\lceil \Delta /
(\max\{m^{1/3}n^{\epsilon/6},n^{\epsilon'} \})^{1/b} \rceil)$ where
$b$ is the doubling dimension of the graph and $\epsilon'$ is any
constant less than $\epsilon$. Observe that for graphs with small
(e.g., constant) doubling dimension, which arise in important
practical realms \cite{Karger02},  the number of rounds can be made
asymptotically much smaller than the diameter if sufficient yet
sublinear local memory is available.  While a similar approach for
diameter estimation has been used in the past in the external-memory
setting (see Section~\ref{previouswork}), the algorithm presented
here, to the best of our knowledge, is the first linear-space
distributed algorithm for the problem requiring a number of parallel
rounds which is sublinear in the diameter.

A very desirable feature of our algorithms is that they lend
themselves to efficient practical implementations.  We report on an
extensive set of experiments conducted on a number of large graphs. A
first batch of experiments shows the effectiveness of our
decomposition strategy in minimizing the maximum cluster radius,
compared against the recent parallel decomposition strategy of
\cite{MillerPX13}, while a second batch shows that the approximation
obtained by our diameter approximation algorithm is in fact much
smaller than the asymptotic bound, even for graphs of unknown doubling
dimension (less than twice the diameter in all tested cases)
and that the algorithm's performance compares very favorably
to the one exhibited by direct competitors such as breadth-first
search and the (almost exact) diameter estimation algorithm HADI
\cite{KangTAFL11}. For graphs of very large diameter, the speed-up
of our algorithm can be of orders of magnitude.

The rest of the paper is organized as
follows. Section~\ref{previouswork} summarizes relevant previous work
on  graph decomposition, $k$-center clustering, and diameter estimation.
Section~\ref{clustering} presents our novel decomposition
and discusses how it can be employed to approximate the $k$-center problem.
Section~\ref{diameter} presents our decomposition-based algorithm
for diameter approximation. Section~\ref{MRimplementation} analyzes
our  strategies in the MR model of \cite{PietracaprinaPRSU12}.
Section~\ref{experiments} reports on  the experimental results, and, finally,
Section~\ref{conclusions} offers some concluding remarks.

\section{Previous work} \label{previouswork}
Parallel clustering algorithms relevant to this work have been
studied in \cite{Cohen98,BlellochGKMPT11,MillerPX13}. In
\cite{Cohen98} the notion of $(\beta,W)$-cover for a weighted graph
$G$ is introduced, which is essentially a decomposition of the graph
into nondisjoint clusters, where each node is allowed to belong to
$\BO{\beta n^{1/\beta} \log n}$ distinct clusters and for any two
nodes at weighted distance at most $W$ in the graph, there is a cluster
containing both. A $(\beta,W)$-cover is obtained by growing clusters
of decreasing radii from successive batches of centers.  The algorithm
presented in \cite{BlellochGKMPT11} is similar but returns disjoint
clusters and guarantees a bound on the average number of edges between
clusters. In \cite{MillerPX13} an alternative clustering algorithm is
proposed which assigns to each node $u \in V$ a random time shift
$\delta_u$, taken from an exponential distribution with parameter
$\beta$, and grows a cluster centered at $u$ starting at time
$\delta_{\rm max}-\delta_u$, where $\delta_{\rm max}$ is the maximum
shift, unless by that time node $u$ has been already covered by some
other cluster. The authors show that in this fashion the graph is
partitioned into clusters of maximum radius $\BO{(\log n) /\beta}$, with
high probability, while the average number of edges between clusters,
hence the size of the quotient graph, is at most $\BO{\beta m}$.  None
of the above clustering approaches guarantees that the maximum radius
of the returned clusters is (close to) minimum with respect to all
possible decompositions of the graph featuring the same number of
clusters.

The related \emph{metric $k$-center} optimization problem requires
that given a set $V$ of $n$ points in a metric space, a subset $M
\subset V$ of $k$ points be found so to minimize the maximum distance
between any $v \in V$ and $M$. The problem is NP-hard even if
distances satisfy the triangle inequality \cite{Vazirani01}, but
polynomial-time 2-approximation sequential algorithms are known
\cite{Gonzalez85,HochbaumS87}. Recently, a constant-approximation
MapReduce algorithm was developed in \cite{EneIM11} under the
assumption that the distances among all $\BT{n^2}$ pairs of points are
given in input. The problem remains NP-hard even if $V$ represents the
set of nodes of an undirected connected graph $G$, and the distance
between two nodes is the length of the shortest path between them. To
the best of our knowledge, no low-depth linear-space parallel strategy
that yields a provably good approximation for this important graph
variant is known in the literature.

In recent years, efficient sequential algorithms for estimating the
diameter of very large graphs have been devised, which avoid the costly
computation of all-pairs shortest paths or the memory-inefficient
transitive closure by performing a limited number of Breadth-First
Searches (BFS) from suitably selected source nodes
\cite{CrescenziGHLM13,ChechikLRSTW14}.  Unfortunately, due to the
inherent difficulty of parallelizing BFS \cite{UllmanY91,KleinS92}
these approaches do not appear amenable to efficient low-depth
parallel implementations.  External-memory algorithms for diameter
estimation which employ a clustering-based strategy similar to ours
have been recently proposed in \cite{Meyer08,AjwaniMV12}. The
algorithm by \cite{Meyer08} basically selects $k$ centers at random
and grows disjoint clusters around the centers until the whole graph
is covered. The author shows that the diameter of the original graph
can be approximated within a multiplicative factor of $\BT{\sqrt{k}
  \log n}$, with high probability, by computing the diameter on the
quotient graph associated with the clustering with suitable edge
weights.  In \cite{AjwaniMV12} a recursive implementation of this
strategy is evaluated experimentally. {This approximation ratio is
competitive with our result only for polylogarithmic values of $k$. However,
observe that  for such small values of $k$ the radius of the $k$ clusters
  must be within a small (polylogarithmic) factor of the graph diameter, and thus the
  parallel number of rounds cannot be substantially sublinear in the
  diameter itself.}

Efficient PRAM algorithms for approximating shortest path distances
between given pairs of nodes are given in \cite{Cohen98,Cohen00}. For
sparse graphs with $m \in \BT{n}$, these algorithms feature
$\BO{n^\delta}$ depth, for any fixed constant $\delta \in (0,1)$, but
incur a polylogarithmic space blow-up due to the use of the
$(\beta,W)$-covers mentioned above.  The algorithms are rather
involved and communication intensive, hence, while theoretically
efficient, in practice they may run slowly when implemented on
distributed-memory clusters of loosely-coupled servers, where
communication overhead is typically high. Moreover, their depth is not
directly related to the graph diameter.

In \cite{PalmerGF02}, an efficient algorithm, called ANF, is devised
to tightly approximate the \emph{neighborhood function} of a graph
$G$, which, for every $t \geq 0$, gives the number of pairs of nodes
at distance at most $t$ in $G$, and, therefore, it can be used to
estimate the diameter.  On a connected graph $G$ of diameter $\Delta$,
ANF executes $\Delta$ iterations and maintains at each node $v$ a
suitable succinct data structure to approximate, at the end of each
Iteration $t$, the number of nodes at distance at most $t$ from $v$.
A MapReduce implementation of ANF, called HADI, has been devised in
\cite{KangTAFL11} using Apache Hadoop. Little experimental evidence of HADI's
performance on large benchmark graphs is available. However, as
confirmed by our experiments (see Section~\ref{experiments}), for
large-diameter graphs HADI's strategy, even if implemented using
faster engines than Hadoop, runs very slowly because of the large
number of rounds and the high communication volume.  A very fast,
multithreaded version of ANF, called HyperANF, has been devised in
\cite{BoldiRV11} for expensive tightly-coupled multiprocessors with
large shared memories, which are not the architectures targeted by our
work.

\section{Clustering algorithm} \label{clustering}

Let $G=(V,E)$ be an undirected connected graph with $n=|V|$ nodes and
$m=|E|$ edges. For any two nodes $u,v \in V$ let $\mbox{dist}(u,v)$
denote the number of edges in the shortest path between $u$ and $v$ in $G$. Also, for any $u \in V$ and
$X \subseteq V$, let $\mbox{dist}(u,X)$ denote the minimum distance
between $u$ and a node of $X$. We now present an algorithm ({\tt
  CLUSTER}) that partitions $V$ into disjoint clusters around suitably
selected nodes called \emph{centers}, so that the radius of each
cluster, defined as the maximum distance of a cluster node from the
center, is small.  As in \cite{Cohen98,BlellochGKMPT11,MillerPX13},
our algorithm activates batches of centers progressively, so to allow
more clusters to cover sparser regions of the graph.  However, unlike
those previous works, we can show that our algorithm minimizes the
maximum cluster radius, within a polylogarithmic factor, a property
that later will turn out crucial for the efficiency of the
diameter-approximation algorithm. A parameter $\tau >0$ is used to
control the size of each batch of activated centers. When two or more
clusters attempt to cover a node concurrently, only one of them,
arbitrarily chosen, succeeds, so to maintain clusters disjoint.  The
algorithm's pseudocode is given below\footnote{Unless explicitly
  indicated, the base of all logarithms is 2.}.

\begin{algorithm}[h]
  \caption{{\tt CLUSTER}$(\tau)$}
\label{alg:cluster}
\begin{small}
\DontPrintSemicolon
\Begin{
\Let{$C$}{$\emptyset$} // {\it (current set of clusters)} \;
\Let{$V'$}{$\emptyset$} // {\it (nodes covered by current set of clusters)} \;
\While{$|V-V'| \geq 8 \tau \log n$}{
Select each node of $V-V'$ as a new center \\
\hspace*{0.3cm} independently with probability $4 \tau \log n/|V-V'|$ \;
Add to $C$ the set of singleton clusters centered at the \\
\hspace*{0.3cm} selected nodes \;
In parallel grow all clusters of $C$ disjointly until \\
\hspace*{0.3cm} $\geq |V-V'|/2$ new nodes are covered \;
\Let{$V'$}{nodes covered by the clusters in $C$} \;
}
{\bf return} $C \cup  \{
             \mbox{singleton clusters centered at the nodes of } V-V' \}$
  }
\end{small}
\end{algorithm}
We define a crucial benchmark for analyzing each iteration of the
algorithm.
\begin{definition}
  Let $k$ be an integer, with $1 \leq k \leq n$, and let $V' \subseteq
  V$ be a subset of at most $n-k$ nodes. We define
  \begin{eqnarray*}
  r(G,V',k)  & = & \min\{r \; : \; \exists \; U \subseteq V-V' \mbox{
    s.t. } |U|=k \;\; \\
    & &  \; \; \; \; \; \wedge  \forall w \in V-V' \;\;
  \dist(w,V'\cup U) \leq r \}.
  \end{eqnarray*}
\end{definition}
Suppose that we have partially grown some clusters covering a subset
$V' \subset V$. We know that by continuing to grow these clusters plus
$k$ new clusters centered in uncovered nodes, $r=r(G,V',k)$ growing
steps are necessary to cover the whole graph.
\noindent
We have:
\begin{theorem}\label{thm1}
  \sloppy For any integer $\tau > 0$, with high probability, {\tt
    CLUSTER}$(\tau)$ computes a partition of $V$ into
  $
  \BO{\tau \cdot \log^2 n}
  $
  disjoint clusters of maximum radius
  \[
  R_{\rm ALG} = \BO{\sum_{i=1}^{\ell} R(i,\tau)}
  \]
  where $\ell = \lceil \log(n/(8\tau \log n)) \rceil$, and
  \[
  R(i,\tau) = \max \{r(G,V',\tau) :
  V' \subset V \wedge |V-V'|=n/2^{i-1} \}
  \]
  for $1 \leq i \leq \ell$.
\end{theorem}
\begin{proof}
  %
  The bound on the number of clusters follows by observing that the
  number of clusters added in each iteration is a binomial random
  variable with expectation $4 \tau \log n$, and that at most $\ell =
  \BO{\log n}$ iterations are executed overall.
  As for the upper bound on $R_{\rm ALG}$, it is
  sufficient to show that in the $i$th iteration, with $i \geq 1$, the
  radius of each cluster (new or old) grows by $\BO{R(i, \tau)}$, with
  high probability.  Let $V_i$ be the set of nodes that at the
  beginning of Iteration $i$ are already covered by the existing
  clusters. By construction, we have that $|V_1| = 0$ and, for $i >
  1$, $|V_i| \geq \sum_{j=1}^{i-1} n/2^j$.  Hence, $|V-V_i| \leq
  n/2^{i-1}$.  Let $R_i = r(G,V_i,\tau)$. It is easy to verify that
  %
  %
  \begin{eqnarray*}
  R_i & \leq & \max_{V' \subset V, |V-V'|=|V-V_i|} r(G,V',\tau)\\
  & \leq & \max_{V' \subset V, |V-V'|=n/2^{i-1}} r(G,V',\tau) \\
  & = & R(i,\tau).
  \end{eqnarray*}
  By definition of $r(G,V_i,\tau)$ we know that
  there must exist $\tau$ nodes, say $u_1, u_2, \ldots, u_{\tau} \in
  V-V_i$, such that each node of $V-V_i$ is at distance at most $R_i$
  from either $V_i$ or one of these nodes. Let us consider the
  partition
  \[
  V-V_i = B_0 \cup B_1 \cup \cdots \cup B_{\tau}
  \]
  where $B_0$ is the set of nodes of $V-V_i$ which are closer to $V_i$
  than to any of the $u_j$'s, while $B_j$ is the set of nodes of
  $V-V_i$ which are closer to $u_j$ than to $V_i$ or to any other
  $u_{j'}$. Let
  \[
  J = \{j \geq 1 \; : \; |B_j| \leq |V-V_i|/(2\tau)\}
  \]
  and note that
  \[
  \sum_{j \in J} |B_j| \leq \tau (|V-V_i|/(2\tau))\leq |V-V_i|/2
  \]
  Therefore, we have that $\sum_{j \not\in J} |B_j|
  \geq |V-V_i|/2$.  Since $|V-V_i| \geq 8 \tau \log n$, it is easy to
  see that for any $j \geq 1$ and $j \not\in J$, in the $i$th
  iteration a new center will be chosen from $B_j$ with probability at
  least $1-1/n^2$. Hence, by the union bound, we conclude that a new
  center will fall in every $B_j$ with $j \geq 1$ and $j \not\in J$,
  with probability at least $1-\tau/n^2$. When this event occurs, $R_i
  \leq R(i, \tau)$ cluster growing steps will be sufficient to reach
  half of the nodes of $V-V_i$ (namely, the nodes belonging to $B_0
  \cup (\cup_{j \not\in J} B_j)$. The theorem follows by applying the
  union bound over the $\ell \leq \log n$ iterations.
\end{proof}

An important issue, which is crucial to assess the efficiency of the
diameter approximation algorithm discussed in the next section, is to
establish how much smaller is the maximum radius $R_{\rm ALG}$ of the clusters
returned by {\tt CLUSTER} with respect to the graph diameter
$\Delta$, which is an obvious upper bound to $R_{\rm ALG}$.
Our analysis will express the relation between  $R_{\rm ALG}$ and
$\Delta$ as a function of the \emph{doubling
  dimension} of the graph, a concept that a number of recent works
have shown to be useful in relating algorithms' performance to graph
properties~\cite{AbrahamCGP10}.

\begin{definition} \label{doublingdim} Consider an undirected graph
  $G=(V,E)$.  The \emph{ball of radius $R$} centered at node $v$ is
  the set of nodes at distance at most $R$ from $v$. Also, the
  \emph{doubling dimension} of $G$ is the smallest integer $b > 0$
  such that for any $R >0$, any ball of radius $2R$ can be covered by
  at most $2^b$ balls of radius $R$.
\end{definition}
The following lemma provides an upper bound on $R_{\rm ALG}$ in terms
of the doubling dimension and of the diameter of the graph $G$.
\begin{lemma} \label{radius-bound} Let $G$ be a connected $n$-node
  graph with doubling dimension $b$ and diameter $\Delta$. For $\tau
  \in \BO{n/\log^2 n}$, with high probability, {\tt CLUSTER}$(\tau)$
  computes a partition of $V$ into $\BO{\tau \cdot \log^2 n}$ disjoint
  clusters of maximum radius
  \[
  R_{\rm ALG} = \BO{\left\lceil {\Delta \over \tau^{1/b}}
    \right\rceil \log n}.
  \]
\end{lemma}
\begin{proof}
Let $R_{\rm OPT}(\tau)$ be the smallest maximum radius
achievable by any decomposition into $\tau$ clusters.  It is easy to see that each $R(i,\tau)$ is a lower bound
to $R_{\rm OPT}(\tau)$, whence $R_{\rm ALG} =  \BO{\sum_{i=1}^{\ell} R(i,\tau)} = \BO{R_{\rm OPT}(\tau) \log n}$.
By iterating the definition of doubling dimension starting from a single ball of radius
$\Delta$ containing the whole graph, one can easily argue that
$G$ can be decomposed into (at most) $\tau$ disjoint clusters of radius
$R = \BO{\lceil \Delta/\tau^{1/b}\rceil}$. The bound on $R_{\rm ALG}$ follows
since $R = \BOM{R_{\rm OPT}(\tau)}$.
\end{proof}
Observe that for graphs with diameter $\Delta = \omega(\log n)$ and
low (e.g., constant) doubling dimension, $R_{\rm ALG}$ becomes
$o(\Delta)$ when $\tau$ is large enough. Indeed, some experimental
work \cite{Karger02} reported that, in practice, big data networks of
interest have low doubling dimension. Also, for applications such as
the diameter estimation discussed in the next section, it is
conceivable that parameter $\tau$ be made as large as $n^\epsilon$,
for some constant $\epsilon > 0$. In fact, the gap between the graph
diameter and $R_{\rm ALG}$ can be even more substantial for irregular
graphs where highly-connected regions and sparsely-connected ones
coexist.  For example, let $G$ consist of a constant-degree expander
of $n-\sqrt{n}$ nodes attached to a path of $\sqrt{n}$ nodes, and set
$\tau = \sqrt{n}$.  It is easy to see that $R(i,\tau) =
\BO{\mbox{poly}(\log n)}$, for $i \geq 1$.  Hence, {\tt
  CLUSTER}$(\tau)$ returns $\sqrt{n} \log^2 n$ clusters of maximum
radius $R_{\rm ALG}=\BO{\mbox{poly}(\log n)}$, which is exponentially
smaller than the $\BOM{\sqrt{n}}$ graph diameter.

\subsection{Approximation to $k$-center}
Algorithm {\tt CLUSTER} can be employed to compute an approximate
solution to the $k$-center problem, defined as follows.  Given an
undirected connected graph $G = (V,E)$ with unit edge weights, a set
$M \subseteq V$ of $k$ \emph{centers} is sought which minimizes the
maximum distance $R_{\rm OPT}(k)$ in $G$ of any node $v \in V$ from
$M$. As mentioned in Section~\ref{previouswork} this problem is
NP-hard. The theorem below states our approximation result.

\begin{theorem} \label{k-center} For $k = \BOM{\log^2 n}$, algorithm
  {\tt CLUSTER} can be employed to yield a $\BO{\log^3
    n}$-approximation to the $k$-center problem with unit edge
  weights, with high probability.
\end{theorem}

\begin{proof}
  Fix $\tau = \BT{k/\log^2 n}$ so that our algorithm returns at most
  $k$ clusters with high probability, and let $M$ be the set of
  centers of the returned clusters. Without loss of generality, we
  assume that $M$ contains exactly $k$ nodes (in case $|M| < k$, we
  can add $k-|M|$ arbitrary nodes to $M$, which will not increase the
  value of objective function). Let $R_{\rm ALG}$ be the maximum
  radius of the clusters returned by our algorithm. As proved in
  Lemma~\ref{radius-bound}, we have that, with high probability
  $R_{\rm ALG} = \BO{R_{\rm OPT}(\tau)
  \log n}$. We conclude the proof by arguing that $R_{\rm OPT}(\tau)
=\BO{R_{\rm OPT}(k) \log^2 n}$. Consider the optimal solution to the
$k$-center problem on the graph, and the associated
clustering of radius $R_{\rm OPT}(k)$. Let $T$ be a spanning tree of
the quotient graph associated with this clustering.  It is easy to see
that $T$ can be decomposed into $\tau$ subtrees of height $\BO{\log^2
  n}$ each. Merge the clusters associated with the nodes of each such
subtree into one cluster and pick any node as center of the merged
cluster.  It is easy to see that every node in the graph is at
distance $D = \BO{R_{\rm OPT}(k) \log^2 n}$ from one of the picked
nodes.  Since $D \geq R_{\rm OPT}(\tau)$, we conclude that $R_{\rm
  OPT}(\tau) =\BO{R_{\rm OPT}(k) \log^2 n}$, and the theorem follows.
\end{proof}

\subsection{Extension to disconnected graphs}
Let $G$ be an $n$-node graph with $h>1$ connected components. It is
easy to see that for any $\tau \geq h$, algorithm {\tt
  CLUSTER$(\tau)$} works correctly with the same guarantees stated in
Theorem~\ref{thm1}.  Also, observe that for $k \geq h$, the $k$-center
problem still admits a solution with noninfinite radius. Given $k \geq
h$, we can still get a $\BO{\log^3 n}$-approximation to $k$-center on
$G$ as follows. If $k = \BOM{h \log^2n}$ we simply run {\tt
  CLUSTER$(\tau)$} with $\tau = \BT{k / \log^2 n}$ as before. If
instead $h \leq k = o(h \log^2n)$ we run {\tt CLUSTER$(h)$} and then
reduce the number of clusters $W = \BO{h \log^2 n}$ returned by the
algorithm to $k$ by using the merging technique described in the proof
of Theorem~\ref{k-center}. It is easy to show that the approximation ratio is
still $\BO{\log^3 n}$.

\section{Diameter estimation} \label{diameter}

Let $G$ be an $n$-node connected graph.  As in  \cite{Meyer08},
we approximate the diameter of $G$ through the diameter of the
quotient graph associated with a suitable clustering of $G$. For the
distributed implementation discussed in the next section, the
clustering will be made sufficiently coarse so that the diameter of
the quotient graph can be computed on a single machine. In order to
ensure a small approximation ratio, we need a refined clustering
algorithm ({\tt CLUSTER2}), whose pseudocode is given in Algorithm~\ref{alg:cluster2}, which
imposes a lower bound on the number of growing steps applied to each
cluster, where such a number is precomputed using the clustering
algorithm from Section~\ref{clustering}.
\begin{algorithm}[h]
\caption{{\tt CLUSTER2}$(\tau)$}
\label{alg:cluster2}
\begin{small}
\DontPrintSemicolon
\Begin{
Run {\tt CLUSTER}$(\tau)$ and let $R_{\rm ALG}$ be the maximum radius of \\
\hspace*{0.3cm} the returned clusters \;
\Let{$C$}{$\emptyset$} // {\it (current set of clusters)} \;
\Let{$V'$}{$\emptyset$} // {\it (nodes covered by current set of clusters)} \;
\For{$i\leftarrow 1$ \KwTo $\log n$}{
Select each node of $V-V'$ as a new center \\
\hspace*{0.3cm} independently with probability $2^{i}/n$\;
Add to $C$ the set of singleton clusters centered at the \\
\hspace*{0.3cm} selected nodes \;
In parallel grow all clusters of $C$ disjointly for \\
\hspace*{0.3cm} $2R_{\rm ALG}$ steps \;
\Let{$V'$}{nodes covered by the clusters in $C$}
}
{\bf return} C\\
}
\end{small}
\end{algorithm}

 Let $\tau >0$ be an integral parameter. We have:
\begin{lemma}\label{correctness}
  \sloppy
  For any integer $\tau > 0$, with high probability algorithm {\tt
    CLUSTER2}$(\tau)$ computes a partition of $V$ into $\BO{\tau
    \log^4n}$ clusters of radius $R_{\rm ALG2} \leq 2R_{\rm ALG}\log
  n$, where $R_{\rm ALG}$ is the maximum radius of a cluster returned
  by {\tt CLUSTER}$(\tau)$.
\end{lemma}
\begin{proof}
  The bound on $R_{\rm ALG2}$ is immediate.  Let $W$ be the number of
  clusters returned by the execution of {\tt CLUSTER}$(\tau)$ within
  {\tt CLUSTER2}$(\tau)$. By Theorem~\ref{thm1}, we have that $W =
  \BO{\tau \log^2n}$, with high probability. In what follows, we
  condition on this event.  For $\gamma = 4/\log_2 e$, define $H$ as
  the smallest integer such that $2^H/n \geq (\gamma W \log n)/n$, and
  let $t=\log n -H$.  For $0\leq i\leq t$, define the event $E_i =$
  "at the end of Iteration $H+i$ of the {\bf for} loop, at most
  $n/2^i$ nodes are still uncovered".  We now prove that the event
  $\cap_{i=0}^{t} E_i$ occurs with high probability. Observe that
  \begin{eqnarray*}
  \Pr\left(\cap_{i=0}^{t} E_i\right) & = &
  \Pr(E_0)\prod_{i=1}^{t} \Pr(E_i | E_0\cap \cdots \cap E_{i-1}) \\
  & = &
  \prod_{i=1}^{t} \Pr(E_i | E_0\cap \cdots \cap E_{i-1})
  \end{eqnarray*}
  since $E_0$ clearly holds with probability one.  Consider an
  arbitrary $i$, with $0 < i \leq t$, and assume that $E_0\cap \cdots
  \cap E_{i-1}$ holds. We prove that $E_i$ holds with high
  probability.  Let $V_i$ be the set of nodes already covered at the
  beginning of Iteration $H+i$. Since $E_{i-1}$ holds, we have that
  $|V-V_i|\leq n/2^{i-1}$. Clearly, if $|V-V_i|\leq n/2^i$ then $E_i$
  must hold with probability one. Thus, we consider only the case
  \[
  {n \over 2^i} < |V-V_i| \leq {n \over 2^{i-1}}
  \]
  Let $R_i = r(G, V_i, W)$ and observe that since $R_{\rm ALG}$ is the
  maximum radius of a partition of $G$ into $W$ clusters, we have that
  $R_{\rm ALG} \geq R_i$.  By the definition of $R_i$, there exist $W$
  nodes, say $u_1, u_2, \ldots, u_W \in V-V_i$, such that each node of
  $V-V_i$ is at distance at most $R_i$ from either $V_i$ or one of
  these nodes.  Let us consider the partition
  \[
  V-V_i = B_0 \cup B_1 \cup \cdots \cup B_W
  \]
  where $B_0$ is the set of nodes of $V-V_i$
  which are closer to $V_i$ than to any of the $u_j$'s, while $B_j$ is
  the set of nodes of $V-V_i$ which are closer to $u_j$ than to $V_i$
  or to any other $u_{j'}$ (with ties broken arbitrarily). Let
  \[
  J = \{j \geq 1 \; : \; |B_j| \geq |V-V_i|/(2W)\}
  \]
  It is easy to see that $|B_0|+\sum_{j \in J} |B_j| \geq |V-V_i|/2$.
  Since we assumed that $|V-V_i| > n/2^i$, we have that for every
  $B_j$ with $j \in J$,
  \[
  |B_j| \geq {|V-V_i| \over 2W} \geq
  {n \gamma \log n \over 2^{H+i+1}}
  \]
  As a consequence, since $\gamma = 4/\log_2 e$, a
  new center will be chosen from $B_j$ in Iteration $H+i$ with
  probability at least $1-1/n^2$. By applying the union bound we
  conclude that in Iteration $H+i$ a new center will fall in every
  $B_j$ with $j \in J$, and thus the number of uncovered nodes will at
  least halve, with probability at least $1-W/n^2 \geq 1-1/n$.

  By multiplying the probabilities of the $\BO{\log n}$ conditioned
  events, we conclude that event $\cap_{i=0}^{t} E_i$ occurs with high
  probability.  Finally, one can easily show that, with high
  probability, in the first $H$ iterations, $\BO{W \log n}$ clusters
  are added and, by conditioning on $\cap_{i=0}^{t} E_i$, at the
  beginning of each Iteration~$H+i$, $1\leq i\leq t$, $\BO{W\log n}$
  new clusters are added to $C$, for a total of $\BO{W\log^2n} =
  \BO{\tau\log^4n}$ clusters.
\end{proof}

Suppose we run {\tt CLUSTER2} on a graph $G$, for some $\tau \in
\BO{n/\log^4 n}$, to obtain a set $C$ of clusters of maximum radius
$R_{\rm ALG2}$. Let $G_C$ denote the quotient graph associated with
the clustering, where the nodes correspond to the clusters and there
is an edge between two nodes if there is an edge of $G$ whose
endpoints belong to the two corresponding clusters. Let $\Delta_C$ be
the diameter of $G_C$. We have:

\begin{theorem}\label{segment}
If  $\Delta$ is the true diameter of $G$, then
$\Delta_C = \BO{\left(\Delta/R_{\rm ALG}\right) \log^2 n)}$,
with high probability.
\end{theorem}
\begin{proof}
Let us fix an arbitrary pair of distinct nodes and an arbitrary
shortest path $\pi$ between them.  We show that
at most $\BO{\lceil |\pi|/R_{\rm ALG} \rceil\log^2 n}$ clusters
intersect  $\pi$ (i.e., contain nodes of $\pi$),
with high probability.  Divide $\pi$ into {\em segments\/} of length
$R_{\rm ALG}$, and consider one such segment $S$.  Clearly, all
clusters containing nodes of $S$ must have their centers belong to
nodes at distance at most $R_{\rm ALG2}$ from $S$ (i.e., distance at
most $R_{\rm ALG2}$ from the closest node of $S$).  Recall that
$R_{\rm ALG2} \leq 2R_{\rm ALG} \log n$. For $1 \leq j \leq 2\log n$,
let $C(S,j)$ be the set of nodes whose distance from $S$ is between
$(j-1)R_{\rm ALG}$ and $jR_{\rm ALG}$, and observe that any cluster intersecting
$S$ must be centered at a node belonging to one of the $C(S,j)$'s. We
claim that, with high probability, for any $j$, there are $\BO{\log
  n}$ clusters centered at nodes of $C(S,j)$ which may intersect $S$.  Fix
an index $j$, with $1 \leq j \leq 2\log n$, and let $i_j$ be the first
iteration of the for loop of algorithm {\tt CLUSTER2} in which some
center is selected from $C(S,j)$. It is easy to see that, due to the smooth
growth of the center selection probabilities,  the number
of centers selected from $C(S,j)$ in Iteration $i_j$ and in Iteration
$i_j+1$ is $\BO{\log n}$, with high probability. Consider now a center
$v$ (if any) selected from $C(S,j)$ in some Iteration $i>i_j+1$. In
order to reach $S$, the cluster centered at $v$ must grow for at least
$(j-1)R_{\rm ALG}$ steps. However, since in each iteration active clusters
grow by $2R_{\rm ALG}$ steps, by the time the cluster centered at $v$
reaches $S$, the nodes of $S$ have already been reached
and totally covered by clusters whose centers have been selected from
$C(S,j)$ in Iterations $i_j$ and $i_j+1$ or, possibly, by some other clusters
centered outside $C(S,j)$.  In conclusion, we have that the nodes of
segment $S$ will belong to $\BO{\log^2 n}$ clusters, with high
probability.  The theorem follows by applying the union bound over all
segments of $\pi$, and over all pairs of nodes in $G$.
\end{proof}

Let $\Delta' = 2R_{\rm ALG2}\cdot(\Delta_C+1)+\Delta_C$.  It is easy
to see that $\Delta_C \leq \Delta \leq \Delta'$. Moreover, since
$R_{\rm ALG2} \leq 2R_{\rm ALG} \log n$ and $R_{\rm ALG2} =
\BO{\Delta}$, we have from Theorem~\ref{segment} that $\Delta' =
\BO{\Delta \log^3 n}$.  The following corollary is immediate.
\begin{corollary} \label{diam-thm} Let $G$ be an $n$-node connected
  graph with diameter $\Delta$. Then, the
  clustering returned by {\tt CLUSTER2} can be used to compute two
  values $\Delta_C, \Delta'$ such that $\Delta_C \leq \Delta \leq
  \Delta' = \BO{\Delta \log^3 n}$, with high probability.
\end{corollary}
\noindent
In order to get a tighter approximation, as in \cite{Meyer08}, after
the clustering we can compute the diameter $\Delta'_C$ of the
following weighted instance of the quotient graph
$G_C=(V_C,E_C)$. Specifically, we assign to each edge $(u,v) \in E_C$
a weight equal to the length of the shortest path in $G$ that connects
the two clusters associated with $u$ and $v$ and comprises only nodes
of these two clusters.  It is easy to see that $\Delta'' = 2R_{\rm
  ALG2}+\Delta'_C$ is an upper bound to the diameter $\Delta$ of $G$,
and $\Delta'' \leq \Delta'$.

It is important to remark that while in \cite{Meyer08} the
approximation factor for the diameter is proportional to the square
root of the number of clusters, with our improved clustering strategy
the approximation factor becomes independent of this quantity, a fact
that will also be confirmed by the experiments.  As we will see in the
next section, the number of clusters, hence the size of the quotient
graph, can be suitably chosen to reduce the complexity of the
algorithm, based on the memory resources.

\sloppy
As a final remark, we observe that the proof of Theorem~\ref{segment} shows
that for any two nodes $u,v$ in $G$ their distance $d(u,v)$ can be
upper bounded by a value $d'(u,v) = \BO{d(u,v) \log^3 n+R_{\rm
    ALG2}}$.  As a consequence, by running {\tt CLUSTER2}$(\tau)$ with
$\tau = \BO{\sqrt{n} / \log^4 n}$ and computing the $\BO{n}$-size
all-pairs shortest-path matrix of the (weighted) quotient graph $G_C$
we can obtain a linear-space distance oracle for $G$ featuring the
aforementioned approximation quality, which is polylogarithmic for
farther away nodes (i.e., nodes at distance $\BOM{R_{\rm ALG2}}$.

\section{Distributed implementation and performance analysis} \label{MRimplementation}

We now describe and analyze a distributed implementation of the
clustering and diameter-approximation algorithms devised in the
previous sections, using the MapReduce (MR) model introduced in
\cite{PietracaprinaPRSU12}. The MR model provides a rigorous
computational framework based on the popular MapReduce paradigm
\cite{DeanG08}, which is suitable for large-scale data processing on
clusters of loosely-coupled commodity servers. Similar models have
been recently proposed in \cite{KarloffSV10,GoodrichSZ11}.
An MR algorithm executes as a
sequence of \emph{rounds} where, in a round, a multiset $X$ of
key-value pairs is transformed into a new multiset $Y$ of pairs by
applying a given reducer function (simply called \emph{reducer} in
the rest of the paper) independently to each subset of pairs of $X$
having the same key.  The model features two parameters $M_G$ and
$M_L$, where $M_G$ is the maximum amount of global memory available
to the computation, and $M_L$ is the maximum amount of local memory
available to each reducer. We use \MR\ to denote a given instance of the model.  The
complexity of an \MR\ algorithm is defined as the number of rounds
executed in the worst case, and it is expressed as a function of the
input size and of $M_G$ and $M_L$. Considering that for big input instances
local and global space are premium resources, the main aim of algorithm design on the model
is to provide strategies exhibiting good space-round tradeoffs for  large ranges of the
parameter values.

The following facts are proved in
\cite{GoodrichSZ11,PietracaprinaPRSU12}.
\begin{fact} \label{prefixsorting} \sloppy The sorting and (segmented)
  prefix-sum primitives for inputs of size $n$ can be performed in
  $\BO{\log_{M_L} n}$ rounds in \MR\ with $M_G=\BT{n}$.
\end{fact}

\begin{fact} \label{matrixmult} \sloppy Two $\ell \times \ell$-matrices
  can be multiplied in $\BO{\log_{M_L} n +\ell^3 /(M_G \sqrt{M_L})}$
  rounds in \MR.
\end{fact}

We can implement the sequence of cluster-growing steps embodied in the
main loops of {\tt CLUSTER} and {\tt CLUSTER2} as a progressive
shrinking of the original graph, by maintaining clusters coalesced
into single nodes and updating the adjacencies accordingly. Each
cluster-growing step requires a constant number of sorting and
(segmented) prefix operations on the collection of edges. Moreover,
the assignment of the original graph nodes to clusters can be easily
maintained with constant extra overhead.  By using
Fact~\ref{prefixsorting}, we can easily derive the following result.
\begin{lemma} \label{clusteringMR} {\tt CLUSTER} (resp., {\tt
    CLUSTER2}) can be implemented in the \MR\ model so that, when
  invoked on a graph $G$ with $n$ nodes and $m$ edges, it requires
  $\BO{R \log_{M_L} m}$ rounds, where $R$ is the total number of
  cluster-growing steps performed by the algorithm. In particular, if
  $M_L = \BOM{n^{\epsilon}}$, for some constant $\epsilon > 0$, the
  number of rounds becomes $\BO{R}$.
\end{lemma}

The diameter-approximation algorithm can be implemented in the \MR\ model
by running {\tt CLUSTER2}$(\tau)$ for a value of $\tau$ suitably
chosen to allow the diameter of the quotient graph to be computed
efficiently. The following theorem shows the space-round tradeoffs
attainable when $M_L$ is large enough.
\begin{theorem} \label{MR-diameter} Let $G$ be a connected graph with
  $n$ nodes, $m$ edges, doubling dimension $b $ and diameter
  $\Delta$. Also, let $\epsilon' < \epsilon \in (0,1)$ be two arbitrary
  constants. On the \MR\ model, with $M_G = \BT{m}$ and $M_L =
  \BT{n^{\epsilon}}$, an upper bound $\Delta' = \BO{\Delta
    \log^3 n}$ to the diameter of $G$ can be computed in
  \[
  \BO{\left \lceil \Delta \log^{4/b} n \over
      \left(\max\{m^{1/3}n^{\epsilon/6},n^{\epsilon'} \}\right)^{1/b}
    \right\rceil \log^2 n}
  \]
  rounds, with high probability.
\end{theorem}

\begin{proof}

Fix $\tau = \BT{n^{\epsilon'}/\log^4 n}$ so that {\tt
  CLUSTER2}$(\tau)$ returns $\BO{n^{\epsilon'}}$ clusters with high
probability.  (In case the number of returned clusters is larger, we
repeat the execution of {\tt CLUSTER2}.)  Let $G_C=(V_C,E_C)$ be
quotient graph associated with the returned clustering. If $|E_C| \leq
M_L$, we can compute the diameter of $G_C$ in one round using a single
reducer. Otherwise, by employing the sparsification technique
presented in \cite{BaswanaS07} we transform $G_C$ into a new graph
$G'_C=(V_C,E'_C)$ with $|E'_C| \leq M_L$, whose diameter is a factor
at most $\BO{\epsilon'/(\epsilon-\epsilon')} = \BO{1}$ larger than the
diameter of $G_C$. The sparsification technique requires a constant
number of cluster growing steps similar in spirit to those described
above, which can be realized through a constant number of prefix and
sorting operations. Hence, the transformation can be implemented in
$\BO{1}$ rounds in \MR. Once $G'_C$ is obtained, its diameter and the
resulting approximation $\Delta'$ to $\Delta$ can be computed in one
round with a single reducer.  Therefore, by combining the results of
Lemmas~\ref{radius-bound}, \ref{correctness}, and~\ref{clusteringMR},
we have that {\tt CLUSTER}$(\tau)$ runs in $\BO{\lceil \Delta /
  \tau^{1/b}\rceil \log n}$ rounds, and {\tt CLUSTER2}$(\tau)$ runs
$\BO{\lceil \Delta / \tau^{1/b}\rceil \log^2 n}$ rounds.  Hence, we
have that the total number of  rounds for computing $\Delta'$
is $\BO{\lceil \Delta \log^{4/b}n/n^{\epsilon'/b} \rceil \log^2 n}$.
Alternatively, we can set $\tau = \BO{\min \{n,
  m^{1/3}n^{\epsilon/6}/\log^4 n\}}$ so to obtain a quotient graph
$G_C$ with $|V_C| = \BT{m^{1/3}n^{\epsilon/6}}$ nodes. We can compute
the diameter of the quotient graph by repeated squaring of the
adjacency matrix. By applying the result of Fact~\ref{matrixmult} with
$\ell = |V_C|$ and observing that $|V_C|^3 = \BO{m \cdot
  n^{\epsilon/2}}=\BO{M_G \sqrt{M_L}}$, we conclude that the
computation of the quotient graph diameter requires only an extra
logarithmic number of rounds.  In this fashion, the total number of
rounds for computing $\Delta'$ becomes $\BO{\lceil \Delta
  \log^{4/b}n/(m^{1/3}n^{\epsilon/6})^{1/b} \rceil \log^2 n}$.  The
theorem follows by noting that for both the above implementations, the
quality of the approximation is ensured by Corollary~\ref{diam-thm}.
\end{proof}

We remark that while the upper bound on the approximation factor is
independent of the doubling dimension of the graph, the round
complexity is expressed as a function of it.  This does not restrict
the generality of the algorithm but allows us to show that for graphs
with small doubling dimension, typically graphs with low expansion,
the number of rounds can be made substantially smaller than the graph
diameter and, in fact, this number decreases as more local memory is
available for the reducers, still using linear global space.  This
feature represents the key computational advantage of our algorithm
with respect to other linear-space algorithms, that, while yielding
tighter approximations, require $\BOM{\Delta}$ rounds.

\section{Experimental results} \label{experiments}
We tested our algorithms on a cluster of 16 hosts, each equipped with
a 12 GB RAM and a 4-core I7 processor, connected by a 10 gigabit
Ethernet network. The algorithms have been implemented using Apache
Spark~\cite{Spark}, a popular engine for
distributed large-scale data processing. We performed tests on several
large graphs whose main characteristics are reported in
Table~\ref{tab:datasets}.  The first graph is a symmetrization of a
subgraph of the Twitter network obtained from the LAW website
\cite{LAW}. The next four graphs are from the Stanford Large Network
Datasets Collection \cite{SNAP} and represent, respectively, the
Livejournal social network and three road networks.  The last graph is
a synthetic $1000 \times 1000$ mesh, which has been included since its
doubling dimension is known, unlike the other graphs, and constant
($b=2$), hence it is an example of a graph where our algorithms are
provably effective.

\subsection{Experiments on the Clustering Algorithm}
We compared the quality of the clustering returned by algorithm {\tt
  CLUSTER} (see Section~\ref{clustering}) against that of the
clustering returned by the algorithm presented in \cite{MillerPX13}
and reviewed in Section~\ref{previouswork}, which, for brevity, we
call {\tt MPX}. Recall that {\tt CLUSTER} uses a parameter $\tau$ to
control the number of clusters, while MPX uses (an exponential
distribution of) parameter $\beta$ to decide when nodes are possibly
activated as cluster centers, hence indirectly controlling the number
of clusters.  Both algorithms aim at computing a decomposition of the
graph into clusters of small radius, so we focused the experiments on
comparing the maximum radius of the returned clusterings. However,
since the minimum maximum radius attainable by any clustering is a
nonincreasing function of the number of clusters, but neither
algorithm is able to precisely fix such a number a priori, we
structured the experiments as follows.

We aimed at decomposition granularities (i.e., number of clusters)
which are roughly three orders of magnitude smaller than the number of
nodes for small-diameter graphs, and roughly two orders of magnitude
smaller than the number of nodes for large-diameter graphs.  We ran
{\tt MPX} and {\tt CLUSTER} setting their parameters $\beta$ and
$\tau$ so to obtain a granularity close enough to the desired one, and
compared the maximum cluster radius obtained by the two algorithms. In
order to be conservative, we gave {\tt MPX} a slight advantage setting
$\beta$ so to always yield a comparable but larger number of clusters
with respect to {\tt CLUSTER}.

Table~\ref{tab:comparison-miller} shows the results of the experiments
for the benchmark graphs. Each row reports the graph, and, for each
algorithm, the number of nodes ($n_C$) and edges ($m_C$) of the
quotient graph associated with the clustering, and the maximum cluster
radius ($r$). The table provides a clear evidence that our algorithm
is more effective in keeping the maximum cluster radius small,
especially for graphs of large diameter. This is partly due to the
fact that {\tt MPX} starts growing only a few clusters, and before
more cluster centers are activated 
the radius of the initial clusters is already grown large. On
the other hand, MPX is often more effective in reducing the number of
edges of the quotient graph, which is in fact the main objective of
the MPX decomposition strategy. This is particularly evident for the
first two graphs in the table, which represent social networks, hence
feature low diameter and high expansion (thus, probably, high doubling
dimension). In these cases, the few clusters initially grown by MPX are
able to absorb entirely highly expanding components, thus resulting in
a more drastic reduction of the edges.

\subsection{Experiments on the Diameter-Ap\-prox\-i\-ma\-tion Algorithm}

For the diameter approximation, we implemented a simplified version of
the algorithm presented in Section~\ref{diameter}, where, for
efficiency, we used {\tt CLUSTER} instead of {\tt CLUSTER2}, thus
avoiding repeating the clustering twice. Also, in order to get a
tighter approximation, we computed the diameter of the weighted
variant of the quotient graph as discussed at the end of
Section~\ref{diameter}. We performed three sets of experiments, which
are discussed below.

The first set of experiments aimed at testing the quality of the
diameter approximation provided by our algorithm.  The results of the
experiments are reported in Table~\ref{tab:diameter-approximation}.
For each graph of Table~\ref{tab:datasets} we estimated the diameter
by running our algorithm with two clusterings of different granularities
(dubbed \emph{coarser} and \emph{finer} clustering, respectively)
reporting, in each case, the number of nodes ($n_C$) and edges ($m_C$)
of the quotient graph $G_C$, the approximation $\Delta'$ and the true
diameter $\Delta$\footnote{In fact, in some cases the ``true
diameter'' reported in the table has been computed through
  approximate yet very accurate algorithms and may exhibit some small
  discrepancies with the actual value.}. Since the quotient graphs
turned out to be sufficiently sparse, the use of sparsification
techniques mentioned in Section~\ref{MRimplementation} was not needed.
We observe that in all cases $\Delta'/\Delta < 2$ and, in fact, the
approximation factor appears to decrease for sparse, long-diameter
graphs. Also, we observe that, as implied by the theoretical results,
the quality of the approximation does not seem to be influenced by the
granularity of the clustering. Therefore, for very large graphs, or
distributed platforms where individual machines are provided with
small local memory, one can resort to a very coarse clustering in
order to fit the whole quotient graph in one machine, and still obtain
a good approximation to the diameter, at the expense, however, of an
increased number of rounds, which are needed to compute the
clustering.

With the second set of experiments, we assessed the time performance of
our algorithm against two competitors: HADI~\cite{KangTAFL11}, which
was reviewed in Section~\ref{previouswork} and provides a rather tight
diameter (under)estimation; and Breadth First Search (BFS), which, as
well known, can be employed to obtain an upper bound to the diameter
within a factor two. HADI's original code, available from \cite{HADI},
was written for the Hadoop framework.  Because of Hadoop's known large
overhead, for fairness, we reimplemented HADI in Spark, with a
performance gain of at least one order of magnitude.  As for BFS, we
implemented a simple and efficient version in Spark.
Table~\ref{USvsHADI} reports the running times and the diameter
estimates obtained with the three algorithms where, for our algorithm,
we used the finer clustering granularity adopted in the experiments
reported in Table~\ref{tab:diameter-approximation}. The figures in the
table clearly show that HADI, while yielding a very accurate estimate
of the diameter, is much slower than our algorithm, by orders of
magnitude for large-diameter graphs. This is due to the fact that HADI
requires $\BT{\Delta}$ rounds and in each round the communication
volume is linear in the number of edges of the input graph. On the
other hand BFS, whose approximation guarantee is similar to ours in
practice, outperforms HADI and, as expected, is considerably slower
than our algorithm on large-diameter graphs. Indeed, BFS still
requires $\BT{\Delta}$ rounds as HADI, but its aggregate
communication volume (rather than the per round communication volume)
is linear in the number of edges of the input
graph.

As remarked in the discussion following Lemma~\ref{radius-bound},
a  desirable feature of our strategy is its capability to
adapt to irregularities of the graph topology, which may
have a larger impact on the performance of the other strategies.
In order to provide experimental evidence of this phenomenon,
our third set of experiments reports the running times
of our algorithm and BFS on three variants of the two
small-diameter graphs (livejournal and twitter) obtained by
appending a chain of $c \cdot \Delta$ extra nodes to a
randomly chosen node, with $c=1, 2, 4, 6, 8, 10$, thus increasing the
diameter accordingly, without substantially altering the overall structure of
the base graph. The plots in Figure~\ref{fig:tails}
clearly show that while the running time of our algorithm is
basically unaltered by the modification, that of BFS
grows linearly with $c$, as expected due to the
strict dependence of the BFS number of rounds from the diameter.
A similar behaviour is to be expected with HADI because of the
same reason.

Putting it all together, the experiments support the theoretical
analysis since they provide evidence that the main competitive
advantages of our algorithm, which are evident in large-diameter
graphs, are the linear aggregate communication volume (as in BFS)
coupled with its ability to run in a number of rounds which can be
substantially smaller than $\Delta$.

\section{Conclusions}\label{conclusions}
We developed a novel parallel decomposition strategy for unweighted,
undirected graphs which ensures a tighter control on both the number
of clusters and their maximum radius, with respect to similar previous
decompositions. We employed our decomposition to devise parallel
polylogarithmic approximation algorithms for the $k$-center
problem and for computing the graph diameter. The algorithms use only
linear overall space and, for a relevant class of graphs (i.e., those
of low doubling dimension), their parallel depth can be made
substantially sublinear in the diameter as long as local memories at
the processing nodes are sufficiently large but still asymptotically
smaller than the graph size.

While the improvement of the approximation bounds and the parallel
depth of our algorithms is a natural direction for further research,
the extension of our findings to the realm of weighted graphs is a
another challenging and relevant open problem. We are currently
exploring this latter issue and have devised a preliminary
decomposition strategy that, together with the number clusters and
their weighted radius, also controls their hop radius, which governs
the parallel depth of the computation.


\onecolumn

\begin{center}
\begin{minipage}{0.45\textwidth}
  \vspace{3cm}
  \small
  \captionof{table}{Characteristics of the graphs used in our experiments}
  \label{tab:datasets}
  \centering
  \begin{tabular}{lrrr}
    \toprule
    Dataset & nodes & edges & diameter\\
    \midrule

    twitter & 39,774,960 & 684,451,342 & 16 \\
    \midrule
    livejournal & 3,997,962 & 34,681,189 & 21 \\
    \midrule
    roads-CA & 1,965,206 & 2,766,607 & 849 \\
    \midrule
    roads-PA & 1,088,092 & 1,541,898 & 786 \\
    \midrule
    roads-TX & 1,379,917 & 1,921,660 & 1,054 \\
    \midrule
    mesh1000 & 1,000,000 & 1,998,000 & 1,998 \\

    \bottomrule
  \end{tabular}
\end{minipage}
~
\begin{minipage}{0.5\textwidth}
  \vspace{2.05cm}
  \small
  \captionof{table}{Comparison between the clusterings returned by
    {\tt CLUSTER} and {\tt MPX}. $n_C$ is the number of clusters,
    $m_C$ is the number of edges between clusters, and $r$ is the
    maximum cluster radius.}
  \label{tab:comparison-miller}
  \centering
  \begin{tabular}{l ggg rrr}
    \toprule

    & \multicolumn{3}{>{\columncolor{Gray}}c}{Algorithm {\tt CLUSTER}} &
    \multicolumn{3}{c}{Algorithm {\tt MPX}} \\

    Dataset & $n_C$ & $m_C$ & $r$ & $n_C$ & $m_C$ & $r$ \\
    \midrule

    twitter
    & 40001 & 17216285 & 5 & 41431 & 109348 & 6 \\
    \midrule

    livejournal
    & 4020 & 230326 & 7 & 5796 & 17098 & 9 \\ 
    \midrule

    roads-CA
    & 15038 & 40597 & 31 & 16429 & 34021 & 61 \\
    \midrule

    roads-PA
    & 7710 & 13300 & 30 & 8529 & 18446 & 58 \\
    \midrule

    roads-TX
    & 10653 & 28582 & 30 & 11238 & 23308 & 55 \\
    \midrule

    mesh1000
    & 7641 & 18476 & 34 & 9112 & 25885 & 56 \\

    \bottomrule
  \end{tabular}
\end{minipage}
\end{center}

\vfill

\begin{center}
  \captionof{table}{Diameter approximation returned by our algorithm on the
    benchmark graphs. $\Delta$ is the diameter of the
    graph, $\Delta'$ is the approximation given by the algorithm.}
  \label{tab:diameter-approximation}
  \centering \small
  \begin{tabular}{l gggg rrrr}
    \toprule
    & \multicolumn{4}{>{\columncolor{Gray}}c}{Coarser clustering}
    & \multicolumn{4}{c}{Finer clustering} \\

    Dataset
    & $n_C$  & $m_C$ & $\Delta'$ & $\Delta$
    & $n_C$  & $m_C$ & $\Delta'$ & $\Delta$ \\
    \midrule

    twitter
    & 1835 & 18865 & 23 & 16   
    & 5276 & 895356 & 27 & 16 \\ 
    \midrule

    livejournal
    & 1933 & 24442 & 29 & 21    
    & 7837 & 570608 & 29 & 21 \\ 
    \midrule

    roads-CA
    & 1835 & 5888 & 1504 & 849     
    & 3863 & 10946 & 1477 & 849 \\ 
    \midrule

    roads-PA
    & 1087 & 3261 & 1240 & 786     
    & 4286 & 12314 & 1245 & 786 \\  
    \midrule

    roads-TX
    & 1316 & 3625 & 1568 & 1054  
    & 3821 & 10880 & 1603 & 1054 \\ 
    \midrule

    mesh1000
t    & 880 & 3224 & 2128 & 1998  
    & 3588 & 14198 & 2014 & 1998 \\ 

    \bottomrule
  \end{tabular}
\end{center}

\vfill

\begin{center}
\begin{minipage}{0.47\textwidth}
  \captionof{table}{Comparison of our approach ({\tt CLUSTER}) with
    {\tt HADI} and {\tt BFS}. Numbers in parentheses are the estimated
    diameter $\Delta'$. Column $\Delta$ reports the original diameter.}
  \label{USvsHADI}
  \centering
  \small
  \begin{tabular}{l rrr r}
    \toprule
    & \multicolumn{3}{c}{Time ($\Delta'$)} & \\

    Dataset
    & {\tt CLUSTER} & {\tt BFS} & {\tt HADI} & $\quad\Delta$ \\
    \midrule

    twitter                
    & 303 (27) & 144 (22) & 3697 (14) & 16 \\
    \midrule

    livejournal                 
    & 113 (29) & 123 (26) & 388 (26) & 21 \\
    \midrule

    roads-CA                  
    & 742 (1477) & 5796 (1418) & 11008 (838) & 849 \\
    \midrule

    roads-PA                  
    & 369 (1245) & 5245 (1244) & 10090 (770) & 786 \\
   \midrule

    roads-TX                  
    & 622 (1603) & 5844 (1466) & 12572 (998) & 1054 \\
    \midrule

    mesh1000                    
    & 373 (2014) & 8627 (2224) & 17287 (1998) & 1998 \\

    \bottomrule
  \end{tabular}
\end{minipage}
\hfill
\begin{minipage}{0.47\textwidth}
  \captionof{figure}{Performance of {\tt CLUSTER} and {\tt BFS} on
    graphs with small variations.}
  \label{fig:tails}
  \centering
  \small
  \includegraphics[width=\columnwidth]{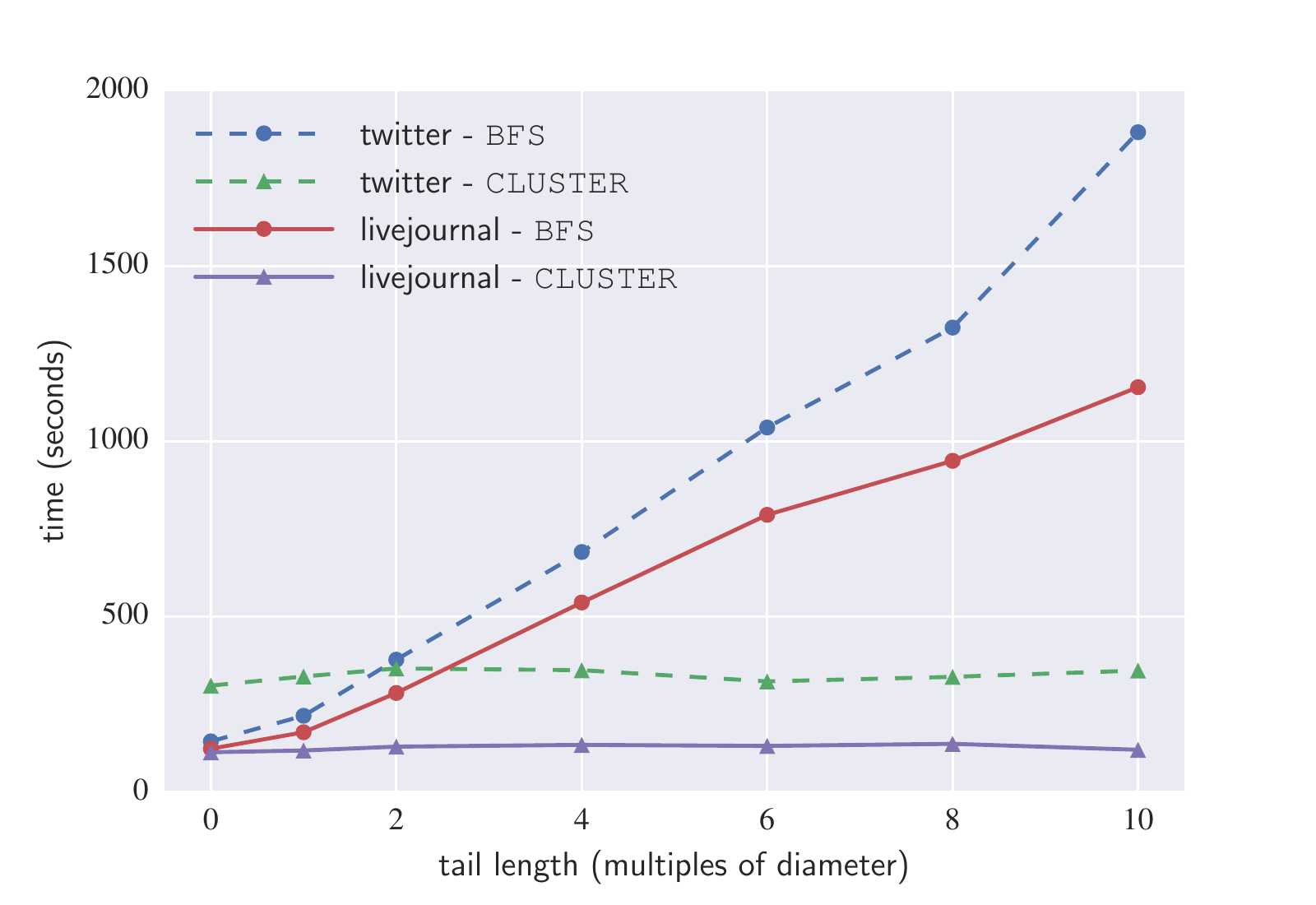}
\end{minipage}
\end{center}

\end{document}